\documentclass[11pt]{article}
 \usepackage[margin=1in]{geometry} 
\usepackage{amsmath,amsthm,amssymb,amsfonts,mathtools,enumitem,tabularx,float,todonotes,array,afterpage}
\usepackage{algorithm,algorithmic}
\usepackage{tikz-cd}

\newcommand{\dist}{\text{dist}}
\newcommand{\Galt}{G_{alt}}
\newcommand{\Valt}{V_{alt}}
\newcommand{\Ealt}{E_{alt}}
\newcommand{\Palt}{P_{alt}}
\newcommand{\Gobs}{G_{obs}}
\newcommand{\Vobs}{V_{obs}}
\newcommand{\Eobs}{E_{obs}}
\newcommand{\Pobs}{P_{obs}}

\def \polylog { \text{\rm polylog~} }
\usepackage{thm-restate}
\usepackage{hyperref}
\usepackage[capitalize,nameinlink]{cleveref}
\crefname{lemma}{Lemma}{Lemmas}
\crefname{claim}{Claim}{Claims}
\usepackage{style_1}

\usepackage[
backend=biber,
style=alphabetic
]{biblatex}
\renewbibmacro{in:}{%
  \ifentrytype{article}{}{\printtext{\bibstring{in}\intitlepunct}}}

\title{Better Lower Bounds for Shortcut Sets and Additive Spanners via an Improved Alternation Product}
\title{Better Lower Bounds for Shortcut Sets and Additive Spanners via an Improved Alternation Product}
\author{Kevin Lu \thanks{Citadel Securities}
\and Virginia Vassilevska Williams \thanks{MIT. Supported by an NSF CAREER Award, NSF Grants CCF-1528078, CCF-1514339 and CCF-1909429, a BSF Grant BSF:2012338, a Google Research Fellowship and a Sloan Research Fellowship.}
\and Nicole Wein \thanks{DIMACS. This work was done while the author was at MIT supported by NSF Grant CCF-1514339..}
\and Zixuan Xu \thanks{MIT}}

\date{}

\begin{document}

\maketitle
\begin{abstract} \small\baselineskip=9pt 
    We obtain improved lower bounds for additive spanners, additive emulators, and diameter-reducing shortcut sets. Spanners and emulators are sparse graphs that approximately preserve the distances of a given graph. A shortcut set is a set of edges that when added to a directed graph, decreases its diameter. The previous best known lower bounds for these three structures are given by Huang and Pettie \cite{Huang2018LowerBO}. For $O(n)$-sized spanners, we improve the lower bound on the additive stretch from $\Omega(n^{1/11})$ to $\Omega(n^{2/21})$. For $O(n)$-sized emulators, we improve the lower bound on the additive stretch from $\Omega(n^{1/18})$ to $\Omega(n^{1/16})$. For $O(m)$-sized shortcut sets, we improve the lower bound on the graph diameter from $\Omega(n^{1/11})$ to $\Omega(n^{1/8})$.
    
    Our key technical contribution, which is the basis of all of our bounds, is an improvement of a graph product known as an \emph{alternation product}. 
   
\end{abstract}

\pagenumbering{arabic}

\section{Introduction}\label{sec:intro}

Given a graph $G$, consider the problem of constructing a graph $H$ that preserves some information about $G$, such as distances or reachability, while optimizing for some other property of $H$. For example, a \emph{spanner} is a subgraph $H$ of $G$ that approximately preserves distances while being as sparse as possible. Another example is a \emph{shortcut set}, which is a (small) set of edges that when added to a directed graph $G$ produces a graph $H$ which preserves the reachability structure of $G$ while reducing the diameter as much as possible. In these examples, and in most other examples, the purpose of the desired property of $H$ (e.g. sparsity and small diameter) is that these properties make it easier to perform computation on $H$ than on $G$. In this case, when faced with a problem on a graph $G$ it can be useful to first compute such a graph $H$ and then solve the problem on $H$. Then, if $H$ preserves the appropriate information about $G$, the solution computed on $H$ will inform the solution for $G$. 

In addition to spanners and shortcut sets, there are many examples of such structures in the literature such as emulators, distance preservers, distance emulators, reachability preservers, transitive closure spanners, diameter and eccentricity spanners, and hopsets. 

\paragraph{Additive Spanners and Emulators.}

Spanners and emulators are widely studied structures that aim to answer the question ``how much can graphs be compressed without compromising too much distance information?" Spanners were first introduced in 1989 by Peleg and Sch\"{a}ffer \cite{peleg1989graph}, and emulators were first brought up by Althofer, Das, Dobkin, Joseph and Soares \cite{ADDJS_steiner_spanner93} as ``Steiner spanners" before being formally introduced by Dor, Halperin, and Zwick \cite{dor2000all}. 
Spanners and emulators have a variety of applications including routing schemes, parallel and distributed algorithms, and distance oracles. See \cite{ahmed2020graph} for a recent survey on all types of spanners and emulators and their many applications. There is also a very recent experimental study comparing different spanner constructions \cite{chimani2021spanner}.

Given an undirected unweighted graph $G=(V,E)$, an \emph{additive spanner} with \emph{stretch} $\beta$ is a subgraph $H\subseteq G$ such that for all $u,v\in V$,

\[\dist_G(u,v)\leq \dist_H(u,v)\leq \dist_G(u,v) + \beta,\]
where $\beta$ can depend on $n$. An \emph{emulator} is the same as a spanner except it can be any (weighted) graph i.e. not necessarily a subgraph of $G$. When constructing additive spanners and emulators, the goal is to obtain as sparse of a graph as possible while still getting a good approximation of the distances.

We focus on \emph{additive} spanners and emulators, which only allow for additive error, but in general spanners and emulators can also have multiplicative error or mixed additive/multiplicative error (e.g. very recently \cite{elkin2021ultra}). 

There has been a long line of work on additive spanners and emulators. See \cref{tab:spanner_bounds,tab:emulator_bounds} for the best known upper and lower bounds for additive spanners and emulators, respectively, for various stretches and sparsities.
\afterpage{
\begin{table}[t!]
\centering
{\renewcommand{\arraystretch}{1.2}
\begin{tabular}{|c|c|c|c|}
\hline
{\bf Citation}                               & {\bf Spanner Size }                 & {\bf Additive Stretch} & {\bf Notes }\\ \hline
Aingworth, Chekuri, & $\Tilde{O}(n^{3/2})$ & 2 & \\
Indyk, Motwani \cite{aingworth1999fast} & & &
\\ \hline
Elkin \& Peleg \cite{EP01add2spanner} & $O(n^{3/2})$ & 2 & \\ \hline
Chechik \cite{chechik2013new} & $\Tilde{O}(n^{7/5})$ & 4 & \\ \hline
Baswana, Kavitha, & $O(n^{4/3})$ & 6 & \\
Mehlhorn, Pettie \cite{baswana2010additive} & & &  \\ \hline
Chechik \cite{chechik2013new} & $O(n^{20/17+\epsilon})$ & $O(n^{4/17-3\epsilon/2})$ & $0 \le \epsilon<8/17$ \\ \hline
Bodwin \& Vassilevska W. \cite{BW} & $\Tilde{O}(n^{10/7-\epsilon})$ & $O(n^\epsilon)$ & $15/49 \le \epsilon< 3/7$ \\ \hline
Bodwin \& Vassilevska W. \cite{BW} & $\Tilde{O}(n^{5/4 - 5\epsilon/12})$ & $O(n^\epsilon)$ & $3/13 \le \epsilon\leq 15/49$ \\ \hline
Bodwin \& Vassilevska W. \cite{BW} & $\Tilde{O}(n^{4/3-7\epsilon/9})$ & $O(n^\epsilon)$ & $0 \le \epsilon \le 3/13$ \\ 
\hline \hline
Abboud \& Bodwin \cite{AB} & $O(n^{4/3-\epsilon})$ & $\Omega(n^\delta)$ & $\delta = \delta(\epsilon) $ \\ \hline
Abboud \& Bodwin \cite{AB} & $O(n)$ & $\Omega(n^{1/22})$ & \\ \hline
Huang \& Pettie \cite{Huang2018LowerBO} & $O(n)$ & $\Omega(n^{1/11})$ & \\ \hline
\textbf{New} & $O(n)$ & $\Omega(n^{2/21})$ &  \\ \hline
\end{tabular}
}
\caption{Upper and lower bounds for additive spanners. Entries above the doubled lines are upper bounds, while entries below them are lower bounds. All lower bounds are included, but only the best known upper bounds.}
\label{tab:spanner_bounds}
\end{table}

\begin{table}[h!]
\centering
{\renewcommand{\arraystretch}{1.2}
\begin{tabular}{|c|c|c|c|}
\hline
{\bf Citation}                               & {\bf Emulator Size }                 & {\bf Additive Stretch} & {\bf Notes }\\ \hline
Aingworth, Chekuri, & $O(n^{3/2})$ & 2 & \\
 Indyk, Motwani \cite{aingworth1999fast} &  & &
\\ \hline
Dor, Halperin, Zwick \cite{DHZ} & $\Tilde{O}(n^{4/3})$ & 4 & \\ \hline
Pettie \cite{pettie2009low} & $O(n)$ & $O(n^{1/4})$ & (implicit in \cite{pettie2009low})\\
\hline \hline
Abboud \& Bodwin \cite{AB} & $O(n^{4/3-\epsilon})$ & $\Omega(n^\delta)$ & $\delta = \delta(\epsilon) $ \\ \hline
Abboud \& Bodwin \cite{AB} & $O(n)$ & $\Omega(n^{1/22})$ & \\ \hline
Huang \& Pettie \cite{Huang2018LowerBO} & $O(n)$ & $\Omega(n^{1/18})$ & \\ \hline
\textbf{New} & $O(n)$ & $\Omega(n^{1/16})$ &  \\ \hline
\end{tabular}
}
\caption{Upper and lower bounds for additive emulators. Entries above the doubled lines are upper bounds, while entries below them are lower bounds. All lower bounds are included, but only the best known upper bounds.}
\label{tab:emulator_bounds}
\end{table}
}

First we discuss prior work on additive spanners. For constant additive stretch, spanners with stretches 2, 4, and 6 are known with sizes $\tilde{O}(n^{3/2})$, $\tilde{O}(n^{7/5})$, and $\tilde{O}(n^{4/3})$, respectively \cite{aingworth1999fast,dor2000all,bollobas2005sparse,baswana2010additive,chechik2013new}. From the lower bounds side, Abboud and Bodwin proved that $\tilde{O}(n^{4/3})$ is tight in the sense that that any purely additive spanner of sparsity $\tilde{O}(n^{4/3-\epsilon})$, for any constant $\epsilon >0$  has  at  least  polynomial in $n$ additive  error \cite{AB}. 

On the other side of the trade-off between sparsity and stretch are spanners on $O(n)$ edges. This setting is the focus of this paper. Since $\Omega(n)$ edges are needed to ensure connectivity, both spanners and emulators need a linear number of edges to have finite error. Thus the setting of $O(n)$ edges is very interesting and natural.

Additive spanners with a near-linear number of edges and additive stretch polynomial in $n$ were first presented by Bollob\'as, Coppersmith and Elkin \cite{bollobas2005sparse}. Later Pettie \cite{pettie2009low} showed that every graph has an $O(n)$-size spanner with
additive stretch $\tilde{O}(n^{9/16})$. Bodwin and Vassilevska W. \cite{BW} improved this to $\tilde{O}(n^{3/7+\epsilon})$ for all $\epsilon>0$. From the lower bounds side, Abboud and Bodwin \cite{AB} showed that $O(n)$-size spanners require additive stretch $\Omega(n^{1/22})$. Huang and Pettie \cite{Huang2018LowerBO} improved this to stretch $\Omega(n^{1/11})$. We  further improve this bound to $\Omega(n^{2/21})$. 

There are also known results for the mid-range of the trade-off between sparsity and stretch (see \cref{tab:spanner_bounds}).

Now we turn to emulators.
For emulators, all of the upper bounds for spanners immediately hold but there are also better known upper bounds and some lower bounds. See \cref{tab:emulator_bounds}. 

We study $O(n)$-sized emulators, so we highlight those results here. Emulators of size $O(n)$ have been studied in \cite{pettie2009low,baswana2010additive,BW15,BW,AB,Huang2018LowerBO} and the best known construction has additive stretch $\tilde{O}(n^{1/4})$ (see \cite{pettie2009low,Huang2018LowerBO}). From the lower bounds side, Abboud and Bodwin \cite{AB} showed that $O(n)$-size emulators require additive stretch $\Omega(n^{1/22})$. Huang and Pettie \cite{Huang2018LowerBO} improved the lower bound to $\Omega(n^{1/18})$. We further improve this bound to $\Omega(n^{1/16})$.

This paper and the aforementioned work focuses on existential bounds for spanners and emulators, but there has also been work on the running time of algorithms to construct spanners and emulators (e.g. \cite{woodruff2010additive, knudsen2014additive,knudsen2017additive, AlDhalaan2021FastCO}).

\paragraph{Shortcut Sets.}

Let $G = (V,E)$ be a directed unweighted graph and $G^* = (V, E^*)$ be its transitive closure i.e. $(u,v)\in E^*$ if $v$ is reachable from $u$ in $G$. The \emph{diameter} of $G$ is the maximum over all of the pairwise distances $\dist_G(u,v)$ for $(u,v)\in E^*$. A \emph{shortcut set} is a set of edges (namely ``shortcuts") $E'\subset E^*$ chosen to minimize the diameter of the graph $G' = (V, E\cup E')$. For a graph on $n$ vertices and $m$ edges, the literature has mainly focused on how small the diameter can become after adding $O(n)$ or $O(m)$ shortcuts. The known results for this problem are shown in Table~\ref{tab:shortcut_bounds}. 

\begin{table}[h]
    \centering
    {\renewcommand{\arraystretch}{1.2}
    \begin{tabular}{ |c|c|c| } 
    \hline
      \textbf{Citation} & \textbf{Shortcut Set Size} & \textbf{Diameter}  \\ 
      \hline 
     Folklore & $O(n)$ & $\tilde{O}(\sqrt{n})$ \\ 
     \hline
     Folklore & $O(m)$ & $\tilde{O}(n/\sqrt{m})$ \\ 
     \hline 
     Hesse \cite{Hesse03} & $O(mn^{1/17})$ & $\Omega(n^{1/17})$\\
     \hline
     Huang \& Pettie \cite{Huang2018LowerBO} & $O(n)$ & $\Omega(n^{1/6})$\\
     \hline
     Huang \& Pettie \cite{Huang2018LowerBO} & $O(m)$ & $\Omega(n^{1/11})$\\
     \hline
     \textbf{New} & $O(m)$ & $\Omega(n^{1/8})$\\
     \hline
     
    \end{tabular}
    }
    \caption{Upper bounds and lower bounds on shortcut sets.}
    \label{tab:shortcut_bounds}
\end{table}

Shortcut sets were first explicitly introduced by Thorup \cite{thorup_digraph} in 1992. He conjectured that for every directed graph there is a shortcut set of size $m$ that reduces the diameter to at most $\polylog(n)$.
This conjecture was confirmed for trees and then for planar graphs \cite{yao1982space,chazelle1987computing,thorup_digraph,thorup1995shortcutting}, but refuted for general graphs by
Hesse \cite{Hesse03}. Hesse constructed a graph with $m=\Theta(n
^{19/17})$ edges and diameter $\Theta(n
^{1/17})$ such that reducing the diameter at all 
requires $\Omega(mn^{1/17})$ shortcuts. More generally, he showed there exist graphs
with $m = n^{1+\epsilon}$ edges and diameter $n^\delta$, $\delta = \delta(\epsilon)$, that require $\Omega(n^{2-\epsilon})$
shortcuts to make the
diameter $o(n^\delta)$. See also \cite{abboud2018hierarchy} for an alternative proof of this result. Huang and Pettie \cite{Huang2018LowerBO} improved the diameter lower bound by proving that $O(n)$-size shortcut sets cannot reduce the diameter below $\Omega(n^{1/6})$, and
that $O(m)$-size shortcut sets cannot reduce the diameter below $\Omega(n^{1/11})$. We improve the latter bound to $\Omega(n^{1/8})$.

From the upper bound side, the best known construction for $O(n)$-sized shortcut sets is given by a folklore algorithm that achieves diameter $\tilde{O}(\sqrt{n})$. The algorithm is simply to pick a random set $S$ of $O(\sqrt{n})$ vertices and include shortcuts between all $O(n)$ pairs $u,v\in S$ such that $v$ is reachable from $u$. There is an analogous algorithm for $O(m)$-sized shortcut sets that achieves diameter $\tilde{O}(n/\sqrt{m})$. Whether one can do better than these folklore algorithms is a very interesting open problem.

One notable application of shortcut sets is for the problem of parallel single-source reachability. It is a long-standing open problem to develop an algorithm  for this problem that is simultaneously work-efficient and has low depth (i.e.~time). There was no known work-efficient algorithm for this problem with nontrivial parallelism until quite recently when Fineman \cite{fineman2019nearly} proved that there is an algorithm with $\tilde{O}(m)$ work and $\tilde{O}(n^{2/3})$ depth. The main tool in his algorithm is the use of shortcut sets, which were also previously used in this context by Ullman and Yannakakis \cite{ullman1991high}. Fineman's main technical contribution is an $\tilde{O}(m)$ time algorithm to to construct a $\tilde{O}(n)$-sized shortcut set that reduces the diameter to $\tilde{O}(n^{2/3})$. We note that the above folklore algorithms achieve lower diameter, however it is not known how to construct them in near-linear time. This result of Fineman was subsequently improved by Jambulapati, Liu, and Sidford \cite{liu2019parallel}, again using shortcut sets, who reduced the diameter to $n^{1/2+o(1)}$, nearly matching the diameter achieved by the folklore algorithm. A trade-off between work and depth was then obtained by Cao, Fineman, and Russell \cite{cao2020improved}, again using shortcut sets. These techniques for parallel reachability also extend to the streaming and distributed settings \cite{liu2019parallel}.

A structure that is closely related to a shortcut set is a \emph{transitive closure spanner} (see the survey \cite{raskhodnikova2010transitive}). Like a shortcut set, a transitive closure spanner is a graph that preserves the reachability structure of a directed graph while reducing the diameter as much as possible. The difference is that a shortcut set is a set of edges that is added to the original graph, while a transitive closure spanner is a separate graph whose sparsity is measured by its total number of edges. Most of the work on transitive closure spanners is not directly relevant to the goals of this paper, because either it is computational rather than existential, or studies special classes of graphs.

Another related structure is shortcut sets for undirected graphs. Using a linear number of shortcuts in an undirected graph, one can simply add a star which reduces the diameter to 2. Thus, the interesting parameter regime for this problem on undirected graphs is when only a sublinear number of shortcuts is allowed. This line of work was initiated by Chung and Garey \cite{chung1984diameter} and has been studied in e.g. \cite{alon2000decreasing,demaine2010minimizing,tan2017shortcutting}.

\paragraph{Lower Bound Constructions.}

Lower bound constructions for additive spanners, additive emulators, and shortcut sets use a common set of techniques, where the goal is to construct a graph with a set of ``critical pairs'' of vertices that have the following key properties: 
\begin{enumerate}
\item[(P1)] There are many critical pairs.
\item[(P2)] The path between each critical pair is \emph{unique}. Thus, we unambiguously call each such path a \emph{critical path}.
\item[(P3)] Each pair of critical paths is \emph{nearly disjoint}. (Traditionally this means that two intersecting critical paths can only intersect at either one vertex or one edge.)
\item[(P4)] Each critical path is \emph{long}.
\end{enumerate} Graphs that satisfy these conditions have been independently discovered by Alon~\cite{alon2002testing}, Hesse \cite{Hesse03}, and Coppersmith and
Elkin~\cite{CE}, and also used in several other papers~\cite{Huang2018LowerBO,AB,abboud2018hierarchy}.

Given a ``base graph'' with these properties, one can then take a \emph{product} of two copies of the base graph to obtain a graph with specific properties tailored to a particular application. Several types of products have been studied in the literature such as the \emph{alternation product} (discovered independently by Hesse \cite{Hesse03} and Abboud and Bodwin \cite{AB}, and also used in~\cite{abboud2018hierarchy,Huang2018LowerBO})\footnote{Note that not all of these papers use the terminology ``alternation product''.} and the \emph{obstacle product} (discovered by  \cite{AB} and also used in \cite{abboud2018hierarchy,Huang2018LowerBO,Bod}). We use both the alternation product and the obstacle product in this paper. 

The purpose of an alternation product is to increase the number of critical pairs (property (P1)) compared to the base graph. This comes at the cost of decreasing the length of the critical paths (relative to the number of vertices in the graph). (property (P4)). The core contribution of this paper is an improvement of the alternation product. 

\subsection{Results}

We improve the lower bounds for $O(n)$-sized additive spanners and emulators, and $O(m)$-sized shortcut sets. In particular, we prove the following:

\begin{restatable}[Spanner lower bound]{theorem}{thmspannerLB}
\label{thm:spanner_LB}
There exists an undirected graph $G$ with $n$ vertices such that any spanner of $G$ with $O(n)$ edges has additive stretch $\Omega(n^{2/21})$.
\end{restatable}

\cref{thm:spanner_LB} is an improvement over the bound of $\Omega(n^{1/11})$ by Huang and Pettie \cite{Huang2018LowerBO}.

\begin{restatable}[Emulator lower bound]{theorem}{thmemulatorLB}
\label{thm:emulator_LB}
There exists an undirected graph $G$ on $n$ vertices such that any emulator of $G$ with $O(n)$ edges has additive stretch $\Omega(n^{1/16})$.
\end{restatable}

\cref{thm:emulator_LB} is an improvement over the bound of $\Omega(n^{1/18})$ by Huang and Pettie \cite{Huang2018LowerBO}.

\begin{restatable}[Shortcut lower bound]{theorem}{thmshortcutLB}
\label{thm:shortcut_LB}
There exists a directed graph $G = (V,E)$ with $n$ vertices and $m$ edges such that for any shortcut set $E'$ of size $O(m)$, the graph $G' = (V, E\cup E')$ has diameter $\Omega(n^{1/8})$.
\end{restatable}

\cref{thm:shortcut_LB} is an improvement over the bound of $\Omega(n^{1/11})$ by Huang and Pettie \cite{Huang2018LowerBO}.

Our core contribution towards obtaining these results is an improvement of the alternation product (introduced above).  Recall that the goal of lower bound constructions for additive spanners, additive emulators, and shortcut sets is to obtain a graph with the properties (P1)-(P4) above. Further recall that applying an alternation product to a base graph boosts property (P1) (number of critical pairs) at the expense of property (P4) (length of critical paths). 

We define two new alternation products, and the key property of both is that they \emph{make up for some of this loss in critical path length (P4)}. That is, we define refined alternation products whose critical paths are longer than in a standard alternation product. The base graph that we use is essentially equivalent to that used in \cite{Huang2018LowerBO}, but our definition of the alternation product is improved.

We call the graphs yielded by our two alternation products $\Galt^2$ and $\Galt^3$ respectively. $\Galt^3$ has even longer critical paths than $\Galt^2$, but at the expense of relaxing property (P3) (near-disjointness of critical paths). In particular, the critical paths in $\Galt^2$ intersect on at most one edge, while the critical paths in $\Galt^3$ intersect on at most a path of length 2. The precise properties of $\Galt^2$ and $\Galt^3$ are specified in the following theorems.

\begin{restatable}{theorem}{thmGalttwo}
\label{thm:Galt2}
(Properties of $\Galt^2$) For any $D,r>0\in \Z$, there exists a graph $\Galt^2(D,r) = (\Valt^2, \Ealt^2)$ with a set $\Palt^2\subset \Valt^2\times \Valt^2$ of critical pairs that has the following properties:
\begin{enumerate}
    \item The numbers of vertices, edges, and critical pairs in $\Galt^2$ are:
    \begin{align*}
    |V_{alt}^2| &= \Theta(D(Dr)^3),\\
    |E_{alt}^2| &=\Theta(D(Dr)^3 r^{2/3}),\\
    |P_{alt}^2| &= \Theta ((Dr)^3 r^{4/3}).
    \end{align*}
    \item For any critical pair $(s,t)\in \Palt^2$, the path between $(s,t)$ is unique and has length $2D$. (Thus we refer to the unique paths between critical pairs as critical paths.)
    \item Every pair of intersecting critical paths intersects on either a single vertex or a single edge.
\end{enumerate}
\end{restatable}

\begin{restatable}{theorem}{thmGaltthree}
\label{thm:Galt3}
(Properties of $\Galt^3$) For any $D,r>0\in \Z$, there exists a graph $\Galt^3(D,r) = (\Valt^3, \Ealt^3)$ with a set $\Palt^3\subset \Valt^3\times \Valt^3$ that has the following properties:
\begin{enumerate}
    \item The numbers of vertices, edges, and critical pairs in $\Galt^3$ are:
    \begin{align*}
    |V_{alt}^3| &= \Theta(D(Dr)^4),\\
    |E_{alt}^3| &= \Theta(D(Dr)^4r^{2/3}),\\
    |P_{alt}^3| &= \Theta((Dr)^4r^2).
    \end{align*}
     \item For any critical pair $(s,t)\in \Palt^3$, the critical path between $(s,t)$ is unique and has length $3D$.
    \item Any pair of critical paths can only intersect on a path of length 2, or a single edge, or a single vertex.
    \item Given a path $p$ of length 2 in $\Galt^3$, at most $O(r^{2/3})$ critical paths can contain $p$.
\end{enumerate}
\end{restatable}

\subsection{Organization}
In \cref{sec:technical_overview} we provide a technical overview. In \cref{sec:prelim} we recall preliminary knowledge that will be used in the following sections. In \cref{sec:result} we present our improved alternation products. In \cref{sec:app_shortcuts}, we prove our lower bound for shortcut sets (\cref{thm:shortcut_LB}). In \cref{sec:spanner_emulator} we prove our lower bounds for spanners and emulators (\cref{thm:spanner_LB},\cref{thm:emulator_LB}). In \cref{sec:conclusion}, we conclude by showing that further generalizations of our alternation products do not yield better results for our applications.

\section{Technical Overview}\label{sec:technical_overview}

Recall that the goal of lower bound constructions for additive spanners, additive emulators, and shortcut sets is to obtain a graph with the properties (P1)-(P4) above. Further recall that applying an alternation product to a base graph boosts property (P1) (number of critical pairs) at the expense of property (P4) (length of critical paths).

The core contribution of this paper is an improved alternation product. We begin by describing the standard alternation product.

\paragraph{The Standard Alternation Product.} We will describe the graph $G_{alt}$ obtained by taking a standard alternation product of two copies of the base graph $G_0$ \cite{Hesse03,AB,abboud2018hierarchy,Huang2018LowerBO}. In $G_{alt}$, every vertex represents a \emph{pair} of vertices in the base graph $G_0$. The critical pairs of $G_{alt}$ have the property that if each of $(x,y)$ and $(x',y')$ is a critical pair in $G_0$, then each of $(x,x')$ and $(y,y')$ is a vertex in $G_{alt}$ and $((x,x'),(y,y'))$ is a critical pair in $G_{alt}$. 

$G_{alt}$ is a layered graph whose only edges are between adjacent layers. For a vertex $(x,y)$ in layer $i$ of $G_{alt}$, the definition of the edges incident to $(x,y)$ is conditioned on the parity of $i$. If $i$ is even then $G_{alt}$ contains an edge from $(x,y)$ to each vertex in layer $i+1$ of the form $(z,y)$ such that $(x,z)$ is an edge in $G_0$. If $i$ is odd then $G_{alt}$ contains an edge from $(x,y)$ to each vertex in layer $i+1$ of the form $(x,z)$ such that $(y,z)$ is an edge in $G_0$. 

Traditionally, each vertex in the base graph $G_0$ represents a lattice point in 2-dimensional space\footnote{\cite{Huang2018LowerBO} also obtain a construction with equivalent guarantees in 3-dimensional space.}. That is, each vertex of $G_0$ is of the form $x=(x_1,x_2)$ where $x_1$ and $x_2$ are the two coordinates of the corresponding lattice point. Note that with this lattice point notation, each vertex of $G_{alt}$ is of the form $((x_1,x_2),(y_1,y_2))$. That is, each vertex in $G_{alt}$ has 4 dimensions.

\paragraph{Our First Alternation Product Graph $G_{alt}^2$.} Recall that the goal of our improved alternation products is to produce graphs with longer critical paths than the standard alternation product. The key idea of our first improved alternation product, notated $G_{alt}^2$ is to \emph{decrease the number of dimensions} of each vertex in $G_{alt}$ from 4 to 3 by packing more information into each dimension. With this modification, we pack more information into each vertex, allowing there to be fewer vertices per layer, and thus more layers and longer critical paths.

In particular, each vertex in $G_{alt}^2$ is represented by a triple $(x_1,x_2,x_3)$, and edges are defined as follows: \begin{itemize}\item For even $i$, there is an edge from the vertex $(x_1,x_2,x_3)$ in layer $i$ to the vertex $(y_1,y_2,x_3)$ in layer $i+1$ if there is an edge from the vertex $(x_1,x_2)$ to the vertex $(y_1,y_2)$ in $G_0$. \item For odd $i$, there is an edge from $(x_1,x_2,x_3)$ in layer $i$ to $(x_1,y_2,y_3)$ in layer $i+1$ if there is an edge from the vertex $(x_2,x_3)$ to the vertex $(y_2,y_3)$ in $G_0$. \end{itemize} That is, the even steps affect coordinates 1 and 2, and the odd steps affect coordinates 2 and 3. (Recall that in the standard alternation product, the even steps affected coordinates 1 and 2, and the odd steps affected coordinates 3 and 4.) 

We will show that $G_{alt}^2$ can be applied to get lower bounds for shortcut sets, additive spanners, and additive emulators that improve over prior work.  Next, we describe our second alternation product graph $G_{alt}^3$, which we use to further improve the lower bound for shortcut sets.

\paragraph{Our Second Alternation Product Graph $G_{alt}^3$.} For our second alternation product, notated $G_{alt}^3$ we take the idea of packing more information into each dimension one step further. Instead of considering an alternation product of two copies of $G_0$, we consider an alternation product of \emph{three} copies of $G_0$ (hence the name $G_{alt}^3$). If we were to apply a standard alternation product on three copies of $G_0$, each vertex in $G_{alt}$ would be represented by 6 coordinates (each vertex in $G_{alt}$ represents 3 vertices in $G_0$ each of which has 2 coordinates). In this standard alternation product, edges from layer $i$ to layer $i+1$ would affect different coordinates depending on the value of $i\mod 3$. If $i\equiv 0 \pmod 3$ then the edges from layer $i$ to layer $i+1$ would affect coordinates 1 and 2, if $i\equiv 1\pmod 3$ they would affect coordinates 3 and 4, and if $i\equiv 2\pmod 3$ they would affect coordinates 5 and 6. 

To construct our improved alternation product $G_{alt}^3$ we decrease the number of dimensions from 6 to 4. In particular, each vertex in $G_{alt}^3$ is represented by a quadruple $(x_1,x_2,x_3,x_4)$, and steps with $i\equiv 0\pmod 3$ affect coordinates 1 and 2, steps with $i\equiv 1\pmod 3$ affect coordinates 2 and 3, and steps with $i\equiv 2\pmod 3$ affect coordinates 3 and 4. Again, this packs more information into each vertex, allowing there to be fewer vertices per layer, and thus more layers and longer critical paths. 

As previously stated, the construction of $G_{alt}^3$ relaxes property (P3) (near-disjointness of critical paths): each pair of critical paths in $G_{alt}^3$ intersects on at most a path of length 2. Coping with this increased intersection requires some extra care in the analysis for our application to shortcut sets. 

Unfortunately, the applications to spanners and emulators do not allow critical paths to overlap on a path of length 2, so we do not use $G_{alt}^3$ for these applications. 

\paragraph{Limitations on Further Generalizations.}
We note that for both of our new alternation products $G_{alt}^2$ and $G_{alt}^3$, one could consider packing even more information into even fewer coordinates. The problem with this is that when too much information is packed into too few coordinates, the critical paths lose the important property that they are unique.

Lastly, one can imagine generalizing our ideas further, defining $G_{alt}^k$ to be an alternation product of $k$ copies of $G_0$ where each vertex in $G_{alt}^k$ is represented by $k+1$ coordinates. Then, for $i\equiv j\pmod k$, edges from layer $i$ to layer $i+1$ with would affect coordinates $j+1$ and $j+2$. Such a construction would produce a graph where each pair of critical paths intersects at a path of length at most $k-1$. This amount of path overlap is not an inherent problem for the application to shortcut sets, however as it turns out, the bound on shortcut sets that $G_{alt}^k$ would produce is not better than the bound produced by $G_{alt}^3$. In particular, if one plots the bounds obtained by using $G_{alt}^k$ for all positive integer values of $k$, there is a global maximum at $k=3$. This is why we use $G_{alt}^3$ and do not generalize our construction further.

\section{Preliminaries}\label{sec:prelim}

In this section, we state some relevant definitions and known facts that will be useful in proving our lower bound result. Let $G = (V,E)$ be a graph that is weighted or unweighted, directed or undirected. We use $n$ to denote the number of vertices $|V|$ and $m$ to denote the number of edges $|E|$. In this paper, all graphs given as inputs to our problems are unweighted. We assume all graphs considered in the shortcut set problem are directed and all graphs considered in the spanner and emulator problems are undirected. In addition, we note that an emulator construction can be weighted.

Now we recall some basic information on polytopes and convexity that is crucial to our constructions.

\subsection{Polytopes and Convexity}\label{subsec:polytope}

For a set $\{v_1,\dots, v_n\}\in \R^d$ of vectors, $\sum_{i = 1}^n a_iv_i$ is called a \emph{convex combination} of $\{v_1,\ldots, v_n\}$ if $a_i\geq 0 \in \R$ for all $1\leq i\leq n$ and $\sum_{i = n}a_i = 1$. Then the \emph{convex hull} of the vectors $\{v_1,\ldots, v_n\}$ is the set of convex combinations of these vectors. We note that in this paper we set $d = 2$, but for the sake of generality, we present our constructions with $d$ as a variable.

A \emph{polytope} is the convex hull of some set of vectors in $\R^d$. Equivalently, a polytope can also be defined as the set of feasible solutions of $n$ linear inequalities for $n > d$. For a polytope $P$, the \emph{vertices} of $P$ is the set of points in $P$ such that $d$ inequalities hold as equalities and the \emph{edges} of $P$ is the set of points in $P$ such that $d-1$ inequalities hold as equalities.

In this paper, we use $B_d(r)$ to denote the set of all lattice points within Euclidean distance $r$ of the origin and $V_d(r)$ denotes the vertices of $B_d(r)$. That is, $V_d(r)$ is the set of all lattice points at the corners of the convex hull of $B_d(r)$. The set $V_d(r)$ is also referred as the vertex hull of the polytope defined by $B_d(r)$. We view the elements in $V_d(r)$ as vectors in $\Z^d$ and from now on we will refer to them as vectors instead of vertices in order to distinguish from vertices of graphs.  

Barany and Larman \cite{BL98} proved the following bound on the size of $V_d(r)$.

\begin{lemma}[\cite{BL98}]\label{lem:polytope_vertex_num}
For any constant dimension $d$, $|V_d(r)| =  \Theta\left(r^{d\frac{d-1}{d+1}}\right)$. 
\end{lemma}

Let $V_d^+(r)$ denote the set of lattice points $(v_1,\dots, v_d)\in  V_d(r)$ with $v_1,\dots, v_d\geq 0$. We note that the bound from \cref{lem:polytope_vertex_num} also holds for $V_d^+(r)$ since the $|V_d^+(r)|$ is a constant fraction of $|V_r(d)|$ when $d$ is constant. 

\section{Improved Alternation Product}\label{sec:result}

In this section, we will be presenting our improved definition of the alternation product, which proves \cref{thm:Galt2} and \cref{thm:Galt3}.

\subsection{The Base Graph $G_0$}\label{subsec:base_graph}
We construct a directed unweighted base graph $G_0$, which is essentially the same as the base graph used in \cite{Huang2018LowerBO}, although the technical definitions differ.

\paragraph{Construction of $\mathbf{G_0}$} Given $D,r,d > 0\in \Z$, we construct the base graph $G_0 = (V_0, E_0)$ as follows. (Note that throughout the paper we use the construction for $d = 2$, but for the sake of generality, we present the construction with $d$ as a variable.)
\begin{itemize}
    \item[--] \textbf{Vertex Set $V_0$}: $G_0$ consists of $D+1$ layers with each layer having $(3Dr)^d$ vertices representing all the unique elements of $(\Z/(3Dr)\Z)^d$. We label the vertex representing $x\in (\Z/(3Dr)\Z)^d$ in layer $i$ as $v_x^i$, for $0\leq i \leq D$.
    \item[--] \textbf{Edge Set $E_0$}: We only put edges between adjacent layers. In particular, for all $i=0,\dots,D-1$, we put an edge from $v_x^i$ to $v_y^{i+1}$ if and only if there exists a vector $\gamma \in V_d^+(r)$ such that $x+\gamma\equiv y$ in $(\Z/(3Dr)\Z)^d$. 
    \item[--] \textbf{Critical Pairs $P_0$}: The critical pairs are of the form $(v_x^0, v_{x+D\gamma}^D)$ for all $x\in (\Z/(3Dr)\Z)^d$ and $\gamma\in V_d^+(r)$ and the \emph{critical path} between such a critical pair consists of the vertices $v_x^0, v_{x+\gamma}^1, v_{x+2\gamma}^2,\dots, v_{x+D\gamma}^D$. 
\end{itemize}

\paragraph{Properties of the Base Graph $\mathbf{G_0}$.} 
Note that the critical paths have length $D$, which is also the diameter of $G_0$. We will show that the critical paths between critical pairs are in fact unique and pairwise edge-disjoint.

\begin{lemma}\label{lem:G_0_unique_path}
Let $(v_x^0, v_{x+D\gamma}^D)$ be a critical pair in $G_0$ for some $x\in (\Z/(3Dr)\Z)^d$ and $\gamma\in V_d^+(r)$, then the critical path consisting of the vertices $v_x^0, v_{x+\gamma}^1, v_{x+2\gamma}^2,\dots, v_{x+D\gamma}^D$ is the unique path between $v_x^0$ and $v_{x+D\gamma}^D$.
\end{lemma}

\begin{proof}
Recall that the edges in $G_0$ are between adjacent layers in the form of $(v_x^i, v_{x+\gamma}^{i+1})$ for some $x\in (\Z/(3Dr)\Z)^d$ and $\gamma\in V_d(r)$. Suppose for contradiction that there exists another path $p'$ different from the critical path between $(v_x^0, v_{x+D\gamma}^D)$ consisting of the vertices $v_x^0, v_{x+\gamma_1}^1, v_{x+\gamma_1+\gamma_2}^2,\dots, v_{x+\sum_{i = 1}^D\gamma_i}^D$ for $\gamma_1,\gamma_2,\dots, \gamma_D\in V_d^+(r)$. Then we must have
\[D\cdot \gamma \equiv \sum_{i = 1}^D \gamma_i\]
in the group $(\Z/(3Dr)\Z)^d$. However, since $\gamma_i\in V_d(r)$ for all $1\leq i\leq D$, we have $||\gamma_i||\leq r$ and therefore we must also have $D\cdot \gamma = \sum_{i = 1}^D \gamma_i$ in $\Z^d$ instead of just in $(\Z/(3Dr)\Z)^d$. Specifically, this implies that $\gamma$ is a non-trivial convex combination of the vectors in $V_d^+(r)$. Since $V_d^+(r)$ is a strictly convex set, the only solution to this equation is when $\gamma_1 = \gamma_2 = \dots = \gamma_D = \gamma$, which is exactly the critical path.
\end{proof}

From Lemma~\ref{lem:G_0_unique_path}, we have the following corollary:

\begin{corollary}\label{cor:G_0_properties}
The graph $G_0$ has the following properties:
\begin{enumerate}
    \item Each edge $(u,v)\in E_0$ is used by exactly one critical path.
    \item The edge set $E_0$ is exactly the union of the critical paths.
    \item Any two intersecting critical paths can only intersect at one vertex.
\end{enumerate}
\end{corollary}

Using \cref{lem:polytope_vertex_num}, the numbers of vertices, edges and critical pairs in $G_0$ are: we obtain the following quantities for $G_0$:
\begin{align*}
    |V_0| &= D\cdot (3Dr)^d = O(D\cdot (Dr)^d)\\
    |E_0| &= O(|V_0|\cdot |V_d^+(r)|) = O(D\cdot (Dr)^d|V_d(r)|) = O\left(D\cdot(Dr)^d\cdot r^{d\frac{d-1}{d+1}}\right)\\
    |P_0| &= O((Dr)^d |V_d^+(r)|) = O\left((Dr)^d \cdot r^{d\frac{d-1}{d+1}}\right)
\end{align*}

For the rest of the paper, we will only use $d = 2$.

\subsection{The Alternation Product Graph $G_{alt}^2$}\label{subsec:G_alt_2}
In this section, we define our first new alternation product, which produces the graph $G_{alt}^2$ when applied to two copies of the base graph $G_0$. We prove \cref{thm:Galt2} which details the properties of $G_{alt}^2$. First we recall \cref{thm:Galt2}:
\thmGalttwo*

\paragraph{Construction of $\mathbf{\Galt^2}$}
We construct the alternation product graph $G_{alt}^2 = (V_{alt}^2, E_{alt}^2)$ on two copies of the base graph $G_0$ with $d = 2$. $G_{alt}^2$ and its set $P_{alt}^2$ of critical pairs are defined as follows:
\begin{itemize}
    \item[--] \textbf{Vertex Set $V_{alt}^2$}: The graph $G_{alt}^2$ consists of $2D+1$ layers with each layer having $(3Dr)^3$ vertices representing all the unique elements in $(\Z/(3Dr)\Z)^3$.
    \item[--] \textbf{Edge Set $E_{alt}^2$}: The edges of $G_{alt}^2$ are defined as follows:
        \begin{itemize}
            \item If $i$ is even, we put an edge from $v_{(x_1,x_2,x_3)}^i$ to $v_{(y_1,y_2,y_3)}^{i+1}$ if and only if there exists an edge from $v_{(x_1,x_2)}^{i/2}$ to $v_{(y_1,y_2)}^{(i/2)+1}$ in $G_0$, i.e. if and only if there exists a vector $\gamma\in V_2^+(r)$ such that $(x_1,x_2) + \gamma\equiv (y_1,y_2)$ in $(\Z/(3Dr)\Z)^d$. 
            \item If $i$ is odd, we put an edge between $v_{(x_1,x_2,x_3)}^i$ and $v_{(y_1,y_2,y_3)}^{i+1}$ if and only if there exists an edge from $v_{(x_2,x_3)}^{\floor{i/2}}$ to $v_{(y_2,y_3)}^{\floor{i/2}+1}$ in $G_0$.
        \end{itemize}
        We remark that since the edges between adjacent layers in $G_0$ are identical, the superscripts on the vertices in $G_0$ in the above definitions do not actually make a difference. This notation is only for illustrating how the odd layers corresponds to one copy of the base graph and the even layers corresponds to the other copy.  
    \item[--] \textbf{Critical Pairs $P_{alt}^2$}: For all $x = (x_1,x_2,x_3)\in (\Z/(3Dr)\Z)^3$ and $\gamma = (\gamma_1, \gamma_2), \gamma' = (\gamma_1',\gamma_2')\in V_2^+(r)$, there is a critical pair in $P_{alt}$ of the form 
    \[(v_{(x_1, x_2,x_3)}^0, v_{(x_1+D\gamma_1, x_2+D\gamma_2+D\gamma_1',x_3+D\gamma_2')}^{2D}).\]
    Note that a critical pair can be uniquely identified by an element $x\in (\Z/(3Dr)\Z)^3$ and a pair $(\gamma, \gamma')\in (V_2^+(r))^2$. 
\end{itemize}

\paragraph{Properties of $\mathbf{\Galt^2}$.} In the following, we prove the properties and bounds of $\Galt^2$ stated in \cref{thm:Galt2} by a sequence of claims.

\begin{claim}\label{claim:Galt2_unique_path}
 For every critical pair $(v_x^0, v_y^{2D})\in P_{alt}^2$, there is a unique path from $v_x^0$ to $v_y^{2D}$ in $G_{alt}^2$. 
\end{claim}

\begin{proof}
Let $v_x^0 = v_{x^0}^0, v_{x^1}^1,\dots, v_{x^{2D-1}}^{2D-1}, v_{x^{2D}}^{2D} = v_{y}^{2D}$ be the vertices on the critical path from $v_x^0$ to $v_y^{2D}$, where $x^i = (x_1^i, x_2^i, x_3^i)$ for $0\leq i\leq 2D$. By construction of the critical pairs, there exist $\gamma = (\gamma_1,\gamma_2),\gamma' = (\gamma_1', \gamma_2')\in V_2^+(r)$ such that the critical path satisfies that for all $i$:
\begin{itemize}
    \item[--] if $i$ is even, then $(x_1^i,x_2^i) + \gamma\equiv (x_1^{i+1},x_2^{i+1}) \pmod{3Dr}$;
    \item[--] if $i$ is odd, then $(x_2^i,x_3^i) + \gamma'\equiv (x_2^{i+1},x_3^{i+1}) \pmod{3Dr}$.
\end{itemize}
Now we show that this path is unique. Suppose there exists vectors $a^1,\dots, a^{D-1}$ and $b^1,\dots,b^{D-1}\in V_2^+(r)$ such that there is another path from $v_x^0$ to $v_y^{2D}$ that goes through the vertices $v_x^0 = v_{z^0}^0, v_{z^1}^1,\dots, v_{z^{2D-1}}^{2D-1}, v_{z^{2D}}^{2D} = v_{y}^{2D}$ where $z^i = (z_1^i, z_2^i, z_3^i)$ for $0\leq i\leq 2D$ such that for all $i$
\begin{itemize}
    \item[--] if $i$ is even, then $(z_1^i,z_2^i) + a^{i/2}\equiv (z_1^{i+1},z_2^{i+1}) \pmod{3Dr}$;
    \item[--] if $i$ is odd, then $(z_2^i,z_3^i) + b^{\floor{i/2}}\equiv (z_2^{i+1},z_3^{i+1}) \pmod{3Dr}$.
\end{itemize}
 Then since this is a path from $v_x^0$ to $v_y^{2D}$, we must have 

\begin{align}
    D\cdot \gamma_1 &= \sum_{k = 0}^{D - 1} a_1^k,\label{eq1}\\
    D\cdot \gamma_2 + D\cdot \gamma'_1 &= \sum_{k = 0}^{D - 1} a_2^k + \sum_{k = 0}^{D - 1} b_1^k,\label{eq2}\\
    D\cdot \gamma'_2 &= \sum_{k = 0}^{D - 1} b_2^k.\label{eq3}
\end{align}
Note that \cref{eq1}, \cref{eq2} and \cref{eq3} along with the following equations hold in both $\Z/(3Dr)\Z$ and in $\Z$ since the quantities on both sides are less than $3Dr$.

For the vectors $\alpha = (\alpha_1,\alpha_2)$ in $V_2^+(r)$, since $V_2^+(r)$ is a strictly convex set, notice that if we consider $\alpha_2$ as a function of $\alpha_1$, then it is a concave function on the points in $V_2^+(r)$; and similarly $\alpha_1$ is a concave function of $\alpha_2$ on the points in $V_2^+(r)$. This means that if we consider the function $f(\alpha_1) = \alpha_2$ where $(\alpha_1,\alpha_2)\in V_2^+(r)$ then any $\alpha_1,\alpha_1'\in \Z/(3Dr)\Z$ satisfies
\[f\left(\frac{\alpha_1+\alpha_1'}{2}\right)\geq \frac{f(\alpha_1)+f(\alpha_1')}{2}.\] 
Note that by construction, $f$ is only well-defined on discrete values $x\in \Z/(3Dr)\Z$ such that $(x,y)\in V_2^+(r)$ for some $y\in \Z/(3Dr\Z)$, thus $(\alpha_1+\alpha_1')/2$ might not be a valid input to $f$. However, we can extend $f$ to include the edges connecting the vertices in the vertex hull $V_2^+(r)$ and $f$ is still a concave function.  Applying the inequality, we have
\[\frac{1}{D}\cdot \sum_{k = 0}^{D-1} a_2^k = \frac{1}{D}\sum_{k = 0}^{D-1} f(a_1^k) \leq f\left(\frac{1}{D} \sum_{k = 0}^{D-1} a_1^k\right) = f(\gamma_1) = \gamma_2.\]
Similarly, the function $g(a_2) = a_1$ where $(a_1,a_2)\in V_2^+(r)$ is also concave, so we have
\[\frac{1}{D}\cdot \sum_{k = 0}^{D-1} b_1^k = \frac{1}{D}\sum_{k = 0}^{D-1} g(b_2^k) \leq g\left(\frac{1}{D} \sum_{k = 0}^{D-1} b_2^k\right) = g(\gamma_2') = \gamma_1'.\]
Therefore we have 
\begin{equation}\label{eq4}D\cdot \gamma_2 \geq \sum_{k = 0}^{D - 1} a_2^k,\quad D\cdot \gamma'_1 \geq \sum_{k = 0}^{D - 1} b_1^k.\end{equation}
Hence by \cref{eq2}, we must have
\begin{equation}\label{eq5}D\cdot \gamma_2 = \sum_{k = 0}^{D - 1} a_2^k,\quad D\cdot \gamma'_1 = \sum_{k = 0}^{D - 1} b_1^k.\end{equation}
Therefore by combining \cref{eq1,eq3} and \cref{eq5}, we have
\begin{equation}\label{eq6}D\cdot \gamma = \sum_{k = 0}^{D-1} a^k, \quad D\cdot \gamma' = \sum_{k = 0}^{D-1} b^k.\end{equation}
Since $V_2^+(r)$ is a strictly convex set, the only set of solutions is $a^k = \gamma, b^k = \gamma'$ for all $0\leq k\leq D-1$.
\end{proof}

\begin{claim}\label{claim:Galt2_intersect}
Every pair of intersecting critical paths in $G_{alt}^2$ intersects on either a single vertex or a single edge. 
\end{claim}

\begin{proof}
Suppose for contradiction that the intersection between a pair of critical paths contains two vertices $v_1 = v_x^i, v_2=v_y^j$ without an edge between them. 
That is, $|i-j|\geq 2$.  Without loss of generality, assume $i < j$ and let $i$ be even. Recall that a critical pair can be uniquely identified by a vertex on the path and a pair $(\gamma, \gamma')\in (V_2^+(r))^2$. By  \cref{claim:Galt2_unique_path}, every critical pair has a unique path, so it suffices to show that for any such pair of vertices $v_1,v_2$, we can uniquely determine $\gamma,\gamma'$ associated to the path. 

Let $x = (x_1,x_2,x_3), y = (y_1,y_2,y_3)$. In the path from layer $i$ to layer $j$, we must travel from an even layer $k$ to an odd layer $(k+1)$ exactly $\floor{(j-i)/2}$ times. Thus we know that
\begin{align*}
    y_1 &\equiv x_1 + \ceil{\frac{j-i}{2}}\gamma_1 \pmod{3Dr}\\
    y_3 &\equiv x_3 + \floor{\frac{j-i}{2}}\gamma'_2 \pmod{3Dr}.
\end{align*}
Since $\gamma,\gamma'\in V_2^+(r)$, we have $\ceil{\frac{j-i}{2}}\gamma_1<3Dr$ and $\floor{\frac{j-i}{2}}\gamma'_2 < 3Dr$. Then since all other variables are known, we can uniquely solve for the values of $\gamma_1,\gamma'_2$. Then we can uniquely determine the values of $\gamma_2,\gamma'_1$ since $V_2^+(r)$ is a vertex hull in which no two elements share a coordinate. 
\end{proof}

\subsection{The Alternation Product Graph $G_{alt}^3$}\label{subsec:G_alt_3}

In this section, we define our second new alternation product, which produces the graph $G_{alt}^3$ when applied to three copies of the base graph $G_0$. We prove \cref{thm:Galt3} which details the properties of $G_{alt}^3$. We first recall \cref{thm:Galt3}:
\thmGaltthree*

\paragraph{Construction of $\mathbf{\Galt^3}$} We construct the alternation product graph $G_{alt}^3 = (V_{alt}^3,E_{alt}^3)$ using three copies of the base graph $G_0$ with $d = 2$. $G_{alt}^3$ and its set $P_{alt}^3$ of critical pairs are defined as follows:
\begin{itemize}
    \item[--] \textbf{Vertex Set $V_{alt}^3$}: The graph $G_{alt}^3$ consists of $3D+1$ layers with each layer having $(3Dr)^4$ vertices representing all of the unique elements in $(\Z/(3Dr)\Z)^4$.
    \item[--] \textbf{Edge Set $E_{alt}^3$}: The edges of $\Galt^3$ are defined as follows:
        \begin{itemize}
            \item If $i\equiv 0 \pmod 3$, we put an edge between $v_{(x_1,x_2,x_3,x_4)}^i$ and $v_{(y_1,y_2,y_3,y_4)}^{i+1}$ if and only if there is an edge from $v_{(x_1,x_2)}^{i/3}$ to $v_{(y_1,y_2)}^{(i/3)+1}$ in $G_0$, i.e. if and only if there exists $\gamma\in V_2^+(r)$ such that $(x_1,x_2) + \gamma\equiv (y_1,y_2)$ in $(\Z/(3Dr)\Z)^d$. 
            \item If $i\equiv 1 \pmod 3$, we put an edge between $v_{(x_1,x_2,x_3,x_4)}^i$ and $v_{(y_1,y_2,y_3,y_4)}^{i+1}$ if and only if there is an edge from $v_{(x_2,x_3)}^{\floor{i/3}}$ to $v_{(y_2,y_3)}^{\floor{i/3}+1}$ in $G_0$.
            \item If $i\equiv 2 \pmod 3$, we put an edge between $v_{(x_1,x_2,x_3,x_4)}^i$ and $v_{(y_1,y_2,y_3,y_4)}^{i+1}$ if and only if there is an edge from $v_{(x_3,x_4)}^{\floor{i/3}}$ to $v_{(y_3,y_4)}^{\floor{i/3}+1}$ in $G_0$. 
        \end{itemize}
        We remark that since the edges between adjacent layers in $G_0$ are identical, the superscripts on the vertices in $G_0$ in the above definitions do not actually make a difference. This notation is only for illustrating how the alternation product corresponds to three copies of the base graph.
    \item[--] \textbf{Critical Pairs $P_{alt}^3$}:  For all $x = (x_1,x_2,x_3,x_4)\in (\Z/(3Dr)\Z)^4$ and $\gamma = (\gamma_1, \gamma_2), \gamma' = (\gamma_1',\gamma_2'), \gamma'' = (\gamma_1'', \gamma_2'')\in V_2^+(r)$, there is a critical pair in $P_{alt}^3$ of the form 
    \[(v_{(x_1, x_2,x_3,x_4)}^0, v_{(x_1+D\gamma_1, x_2+D\gamma_2+D\gamma_1',x_3+D\gamma_2'+D\gamma_1'', x_4+D\gamma_2'')}^{3D}).\]
    Note that a critical pair can be uniquely identified by an element $x\in (\Z/(3Dr)\Z)^4$ and a $3$-tuple $(\gamma, \gamma', \gamma'')\in (V_2^+(r))^3$.
\end{itemize}

\paragraph{Properties of $\mathbf{\Galt^3}$} Now we prove the properties and bounds of $\Galt^3$ stated in \cref{thm:Galt3} by the following sequence of claims.

\begin{claim}\label{claim:Galt3_unique_path}
For any critical pair $(v_x^0, v_y^{3D})$ in $P_{alt}^3$, there exists a unique path from $v_x^0$ to $v_y^{3D}$ in $G_{alt}^3$.
\end{claim}

\begin{proof}
Let $v_x^0 = v_{x^0}^0, v_{x^1}^1,\dots, v_{x^{3D-1}}^{3D-1}, v_{x^{3D}}^{3D} = v_{y}^{3D}$ be the vertices on the critical path from $v_x^0$ to $v_{y}^{3D}$, where $x^i = (x_1^i, x_2^i, x_3^i,x_4^i)$ for $0\leq i\leq 3D$. By construction of the critical pairs, there exist $\gamma = (\gamma_1,\gamma_2), \gamma' = (\gamma'_1,\gamma'_2), \gamma'' = (\gamma''_1,\gamma''_2)\in V_2^+(r)$ such that the critical path satisfies for all $i$:
\begin{itemize}
    \item[--] if $i\equiv 0 \pmod 3$, then $(x_1^i,x_2^i) + \gamma\equiv (x_1^{i+1},x_2^{i+1})\pmod{3Dr}$;
    \item[--] if $i\equiv 1 \pmod 3$, then $(x_2^i,x_3^i) + \gamma'\equiv (x_2^{i+1},x_3^{i+1})\pmod{3Dr}$;
    \item[--] If $i\equiv 2 \pmod 3$, then $(x_3^i,x_4^i) + \gamma''\equiv (x_3^{i+1},x_4^{i+1})\pmod{3Dr}$.
\end{itemize}
Now we show that this path is unique. Suppose there exist vectors $a^1,\ldots,a^{D-1}$, $b^1,\dots,b^{D-1}$, and $c^1,\dots,c^{D-1}$ in $V_2^+(r)$ such that there is another path from $v_x^0$ to $v_y^{3D}$ that goes through the vertices $v_x^0 = v_{z^0}^0, v_{z^1}^1,\dots, v_{z^{3D-1}}^{3D-1}, v_{z^{3D}}^{3D} = v_{y}^{3D}$ where $z^i = (z_1^i, z_2^i, z_3^i,z_4^i)$ for $0\leq i\leq 3D$ such that for all $i$
\begin{itemize}
    \item[--] if $i\equiv 0 \pmod 3$, then $(z_1^i,z_2^i) + a^{i/3}\equiv (z_1^{i+1},z_2^{i+1}) \pmod{3Dr}$;
    \item[--] if $i\equiv 1 \pmod 3$, then $(z_2^i,z_3^i) + b^{\floor{i/3}}\equiv (z_2^{i+1},z_3^{i+1}) \pmod{3Dr}$;
    \item[--] If $i\equiv 2 \pmod 3$, then $(z_3^i,z_4^i) + c^{\floor{i/3}}\equiv (z_3^{i+1},z_4^{i+1}) \pmod{3Dr}$.
\end{itemize}
Then we must have the following equations 
\begin{align}
    D\cdot \gamma_1 &= \sum_{k = 0}^{D - 1} a_1^k,\label{eq7}\\
    D\cdot \gamma_2 + D\cdot \gamma'_1 &= \sum_{k = 0}^{D - 1} a_2^k + \sum_{k = 0}^{D - 1} b_1^k,\label{eq8}\\
    D\cdot \gamma'_2 + D\cdot \gamma''_1 &= \sum_{k = 0}^{D - 1} b_2^k + \sum_{k = 0}^{D - 1} c_1^k,\label{eq9}\\
    D\cdot \gamma''_2 &= \sum_{k = 0}^{D - 1} c_2^k.\label{eq10}
\end{align}
Note that \cref{eq7,eq8,eq9,eq10} along with the following equations hold in both $\Z/(3Dr)\Z$ and in $\Z$ since the quantities involved on both sides are less than $3Dr$.

Similar to the the proof of \cref{claim:Galt2_unique_path}, for the vectors $(\alpha_1,\alpha_2)$ in $V_2^+(r)$, since $V_2^+(r)$ is a strictly convex set, notice that if we consider $\alpha_2$ as a function of $\alpha_1$, then it is a concave function on the points in $V_2^+(r)$; and conversely $\alpha_1$ is a concave function of $\alpha_2$ on the points in $V_2^+(r)$. Therefore since we have \cref{eq7,eq10}, we must have
\begin{equation}\label{eq11}
    D\cdot \gamma_2 \geq \sum_{k = 0}^{D - 1} a_2^k,\quad D\cdot \gamma''_1 \geq \sum_{k = 0}^{D - 1} c_1^k.
\end{equation}
Suppose that this inequality is strict, that is, we have
\begin{equation}\label{eq12}
    D\cdot \gamma_2 > \sum_{k = 0}^{D - 1} a_2^k,\quad D\cdot \gamma''_1 > \sum_{k = 0}^{D - 1} c_1^k.
\end{equation}
Then by \cref{eq8,eq9} we must have
\begin{equation}\label{eq13}
    D\cdot \gamma'_1 < \sum_{k = 0}^{D - 1} b_1^k,\quad D\cdot \gamma'_2 < \sum_{k = 0}^{D - 1} b_2^k,
\end{equation}
which would violate the fact that $V_2^+(r)$ is a strictly convex set. Therefore we must have
\begin{equation}\label{eq14}
    D\cdot \gamma_2 = \sum_{k = 0}^{D - 1} a_2^k,\quad D\cdot \gamma''_1 =\sum_{k = 0}^{D - 1} c_1^k,\quad  D\cdot \gamma'_1 = \sum_{k = 0}^{D - 1} b_1^k,\quad D\cdot \gamma'_2 = \sum_{k = 0}^{D - 1} b_2^k.
\end{equation}
Combining \cref{eq7,eq10,eq14}, we have the equations
\begin{equation}\label{eq15}
    D\cdot \gamma = \sum_{k = 0}^{D-1} a^k,\quad D\cdot \gamma' = \sum_{k = 0}^{D-1} b^k, \quad D\cdot \gamma'' = \sum_{k = 0}^{D-1} c^k.
\end{equation}
Since $\{a^k\},\{b^k\},\{c^k\}$ for $0\leq k\leq D-1$ are vectors in $V_2^+(r)$, and $V_2^+(r)$ is a strictly convex set, the only set of solutions is $a^k = \gamma, b^k = \gamma', c^k = \gamma''$ for all $0\leq k\leq D-1$.
\end{proof}

\begin{claim}\label{claim:Galt3_intersect_p2}
Each pair of intersecting critical paths of $G_{alt}^3$ can only intersect on a path of length $2$, or a single edge, or a single vertex.
\end{claim}

\begin{proof}
Suppose for contradiction that the intersection between a pair of critical paths contains two vertices $v_1 = v_x^i, v_2=v_y^j$ such that $|i-j|\geq 3$. Without loss of generality, assume $i < j$ and let $i\equiv 0\pmod 3$. Then by \cref{claim:Galt3_unique_path}, every critical pair has a unique path, so it suffices to show that for any such pair of vertices $v_1,v_2$, we can uniquely determine $\gamma,\gamma', \gamma''\in V_2^+(r)$ associated to the path. 

Let $x = (x_1,x_2,x_3,x_4), y = (y_1,y_2,y_3,y_4)$ and suppose $(j-i)\equiv 0\pmod 3$. Then in the path from layer $i$ to layer $j$, we must travel from a layer $k\equiv 0,1,2\pmod 3$ to a layer $(k+1)\equiv 1,2,0\pmod 3$ respectively exactly $(j-i)/3$ times. For the cases where $(j-i)\equiv 1\pmod 3$, we travel from a layer $k\equiv 0\pmod 3$ to $k+1\equiv 1\pmod 3$ $\ceil{(j-i)/3}$ times and similarly for the case where $(j-i)\mod 3\equiv 2$, so here we describe the case where $(j-i)\equiv 0\pmod 3$ and other cases follow with almost identical argument. In this case we know that
\begin{align*}
    y_1 &\equiv x_1 + \frac{j-i}{3}\gamma_1 \pmod{3Dr}\\
    y_4 &\equiv x_4 + \frac{j-i}{3}\gamma_2'' \pmod{3Dr}.
\end{align*}
Since $\gamma,\gamma''\in V_2^+(r)$, we have $\frac{j-i}{3}\gamma_1<3Dr$ and $\frac{j-i}{3}\gamma_2'' <3Dr$. Then since all other variables except $\gamma_1$ and $\gamma_2''$ are known, we can uniquely solve for the values of $\gamma_1, \gamma_2''$. There are unique elements $\gamma_2, \gamma_1''$ such that $(\gamma_1,\gamma_2), (\gamma_1'',\gamma_2'')$ are elements of $V_2^+(r)$ so we can uniquely solve for $\gamma_2, \gamma_1''$. We also know that
\begin{align*}
    y_2 &\equiv x_2 + \frac{j-i}{3}\gamma_2 + \frac{j-i}{3}\gamma_1' \pmod{3Dr}\\
    y_3 &\equiv x_3 + \frac{j-i}{3}\gamma_1'' + \frac{j-i}{3}\gamma_2'\pmod{3Dr}.
\end{align*}
Similarly, since the values on both sides of the equations are less than $3Dr$ and we already know the values for $\gamma_2,\gamma_1''$, we can uniquely solve for the values $\gamma' = (\gamma_1', \gamma_2')$. Hence we have uniquely determined $\gamma, \gamma', \gamma''\in V_2^+(r)$.
\end{proof}

\begin{claim}\label{claim:Galt3_p2_bound}
Given a path $p$ of length $2$ in $G_{alt}^3$, at most $|V_2^+(r)|$ critical paths can contain $p$ as a subpath.
\end{claim}

\begin{proof}
Let $(v_{\alpha}^i, v_{\beta}^{i+2})$ with $\alpha = (\alpha_1,\ldots, \alpha_4), \beta = (\beta_1,\ldots, \beta_4)$ denote the endpoints of the path $p$ of length $2$. Since the three cases where $i\equiv 0,1,2\pmod 3$ are nearly identical, we assume that $i\equiv 0\pmod 3$. Then for some $\gamma = (\gamma_1, \gamma_2), \gamma' = (\gamma_1', \gamma_2')$ in $V_2^+(r)$, we have that $\alpha,\beta$ must satisfy 
\begin{align*}
    \beta_1 &\equiv \alpha_1 + \gamma_1 \pmod{3Dr},\\
    \beta_2 &\equiv \alpha_2 + \gamma_2 + \gamma_1'\pmod{3Dr},\\
    \beta_3 &\equiv \alpha_3 + \gamma_2' \pmod{3Dr},\\
    \beta_4 &\equiv \alpha_4 \pmod{3Dr}.
\end{align*}
 Since the quantities on both sides of the equations are less than $3Dr$ and we know $\alpha$ and $\beta$, we can uniquely determine $\gamma_1, \gamma_2'$ from the above equations. Then since $\gamma, \gamma'$ are elements of $V_2^+(r)$, we can uniquely determine them given $\gamma_1$ and $\gamma_2'$. 

In addition we know that a critical path starting at $v_x^0$ with $x = (x_1,x_2,x_3,x_4)$ that uses $p$ as a subpath must satisfy
\begin{align*}
    \beta_1 &= x_1 + \left(\frac{i}{3}+1\right)\cdot \gamma_1 \pmod{3Dr},\\
    \beta_2 &= x_2 + \left(\frac{i}{3}+1\right)\cdot \gamma_2 + \left(\frac{i}{3}+1\right)\cdot \gamma_1' \pmod{3Dr},\\
    \beta_3 &= x_3 + \left(\frac{i}{3}+1\right)\cdot \gamma_2' + \frac{i}{3}\cdot \gamma_1'' \pmod{3Dr},\\
    \beta_4 &= x_4 + \frac{i}{3}\cdot \gamma_2''\pmod{3Dr}.
\end{align*}
 Then we can uniquely determine $x_1, x_2$ since we know $\gamma, \gamma', \beta_1,\beta_2$. Note that $\gamma'' = (\gamma_1'', \gamma_2'')\in V_2^+(r)$ and $x_3,x_4\in \Z/(3Dr)\Z$ remain undetermined. However, notice that by choosing any $\gamma''\in V_2^+(r)$, we can uniquely determine $x_3, x_4$ from the above equations involving $\beta_3$ and $\beta_4$. Since a critical path can be uniquely identified by an element $x\in (\Z/(3Dr)\Z)^4$ and a $3$-tuple $(\gamma, \gamma', \gamma'')\in (V_2^+(r))^3$, we can therefore conclude that at most $|V_2^+(r)|$ critical paths can pass through $p$.
\end{proof}

\section{Application to Shortcuts}\label{sec:app_shortcuts}

In this section we will prove \cref{thm:shortcut_LB}, which we recall: \thmshortcutLB* 

We prove \cref{thm:shortcut_LB} using $G_{alt}^3$. As a warm-up we we first show how to obtain a lower bound of $\Omega(n^{2/17})$ using $G_{alt}^2$. This bound is not as good as our final bound of $\Omega(n^{1/8})$, but still better than the current best known bound of $\Omega(n^{1/11})$ \cite{Huang2018LowerBO}. 

It may appear that further generalizing the construction of $\Galt^2$ and $\Galt^3$ to $\Galt^k$ by taking the alternation product of more copies of $G_0$ could lead to further improvement of the lower bound. However, in \cref{sec:conclusion} we will show that in fact $\Galt^3$ gives the best possible bound using this technique.

\paragraph{$\mathbf{\Omega(n^{2/17})}$ Lower Bound from $\mathbf{\Galt^2}$} First we state a lemma relating the size of a shortcut set of $\Galt^2$ with the diameter that it yields.

\begin{lemma}\label{lem:G_alt_2_shortcut_lb}
Let $E'$ be a shortcut set on $G_{alt}^2(D,r)$. If $|E'| < |P_{alt}^2|$ then the diameter of $G' = (V_{alt}^2, E_{alt}^2\cup E')$ is $2D$. 
\end{lemma}

\begin{proof}

Recall that the diameter of $G_{alt}^2$ is $2D$.  
Note that every useful shortcut connects a pair of vertices at distances at least 2. Then, by Lemma~\ref{claim:Galt2_intersect}, each shortcut can only decrease the distance between one critical pair in $P_{alt}^2$. Therefore with fewer than $|P_{alt}^2|$ shortcuts, we cannot reduce the diameter below $2D$.
\end{proof}

We set $|P_{alt}^2|=|E'|+1$ and then solve for the parameters $r$ and $D$, as follows. By assumption, $|E'|=\Theta(m)=\Theta(|\Ealt^2|)$, so $|P_{alt}^2|=\Theta(|\Ealt^2|)$. Thus, by the bounds on $|P_{alt}^2|$ and $|\Ealt^2|$ from \cref{thm:Galt2}, we have \[(Dr)^3 r^{4/3} = \Theta( D(Dr)^3 r^{2/3}).\] Solving for $r$, we get $r = \Theta(D^{3/2})$. Then we have
\[|\Valt^2| = n = \Theta (D(Dr)^3) = \Theta(D^{17/2}).\]
Hence this shows that $D = \Omega(n^{2/17})$.

Since $|P_{alt}^2|=|E'|+1$, Lemma~\ref{lem:G_alt_2_shortcut_lb} implies that a shortcut set of size $O(m)$ cannot reduce the diameter below $2D=\Omega(n^{2/17})$.

\paragraph{$\mathbf{\Omega(n^{1/8})}$ Lower Bound from $\mathbf{\Galt^3}$}

First we state a lemma relating the size of a shortcut set of $\Galt^3$ with the diameter that it yields.

\begin{lemma}\label{lem:Galt3_diameter}
Let $E'$ be a shortcut set for $G_{alt}^3(D,r)$ with $|E'| = O(|E_{alt}^3|)$. If $|E'|< |P_{alt}^3|$ then  the diameter of $G' = (V_{alt}^3, E_{alt}^3\cup E')$ is $3D - O(r^{2/3})$. 
\end{lemma}

\begin{proof}
Recall that the diameter of $\Galt^3$ is $3D$. 

We say that a shortcut has \emph{length} $\ell$ if its two endpoints are of distance $\ell$ in $\Galt^3$.
Note that every useful shortcut must have length at least 2. We consider the contribution to the decrease in graph diameter due to shortcuts of length exactly 2 and shortcuts of length greater than 2 separately. The reason we consider these two cases separately is because by \cref{claim:Galt3_intersect_p2} a shortcut of length exactly 2 can decrease the distance between \emph{multiple} critical paths.
Let $k$ be the number of shortcuts of length exactly 2. So $|E'|-k$ is the number of shortcuts of length greater than $2$.

We first consider shortcuts of length greater than $2$. By Lemma~\ref{claim:Galt3_intersect_p2}, any intersecting pair of critical paths can only intersect on a path of length $2$, a single edge, or a single vertex. Thus, any shortcut of length greater than $2$ can only decrease the length of one critical path. Thus, the number of critical paths that are shortened by shortcuts of length greater than $2$ is at most $|E'|-k$. By assumption, $|E'|<|\Palt^3|$, so the diameter of the union of $\Galt^3$ and the shortcuts of length greater than $2$ is $3D$. 

Now we consider the shortcuts of length exactly 2. By Lemma~\ref{claim:Galt3_p2_bound}, a shortcut of length $2$ can decrease the distance between at most $|V_2^+(r)|$ critical pairs. Note that a shortcut of length 2 can only decrease the distance between each such critical pairs by 1. Thus, the total decrease in distance across all critical pairs by shortcuts of length 2 is at most $k\cdot|V_2^+(r)|$. Since the distances between at most $|E'|-k$ critical pairs are decreased by shortcuts of length greater than $2$, the distances between at least $|\Palt^3|-|E'|+k$ critical pairs are \emph{not} decreased by shortcuts of length greater than $2$. By assumption, $|E'|<|\Palt^3|$ so $|\Palt^3|-|E'|+k>k$. That is, if the diameter decreases below $3D$, this means that the shortcuts of length 2 must decrease the distances between at least $k$ critical pairs.  Then since the total decrease in distance across all critical pairs by shortcuts of length 2 is at most $k\cdot|V_2^+(r)|$, and this decrease in distance is distributed over at least $k$ critical pairs, the diameter can decrease by at most \[\frac{k\cdot|V_2^+(r)|}{k}=|V_2^+(r)|=\Theta(r^{2/3}).\]
Thus, the final diameter is $3D - O(r^{2/3})$.
\end{proof}

We set $|P_{alt}^3|=|E'|+1$ and then solve for the parameters $r$ and $D$, as follows. By assumption, $|E'|=\Theta(m)=\Theta(|\Ealt^3|)$, so $|P_{alt}^3|=\Theta(|\Ealt^3|)$. Thus, by the bounds on $|P_{alt}^3|$ and $|\Ealt^3|$ from \cref{thm:Galt3}, we have \[(Dr)^4 r^{2} = \Theta(D(Dr)^4 r^{2/3}).\] Solving for $r$, we get $r = \Theta(D^{3/4})$. Then we have
\[|\Valt^3| = n = \Theta (D(Dr)^4) = \Theta(D^8).\]
Hence this shows that $D = \Omega(n^{1/8})$.

Since $|P_{alt}^3|=|E'|+1$ and $r^{2/3}=o(D)$, Lemma~\ref{lem:Galt3_diameter} implies that a shortcut set of size $O(m)$ cannot reduce the diameter below $3D-O(r^{2/3})=3D-o(D)=\Omega(n^{1/8})$. 

\section{Applications to Additive Spanners \& Emulators}\label{sec:spanner_emulator}

In this section, we first construct the graph $\Gobs$ by applying the obstacle product on $\Galt^2$ from Subsection~\ref{subsec:Gobs}. Then we show that the graph $\Gobs$ can give lower bound results for additive spanners and emulators. The obstacle product is a technique first introduced by Abboud and Bodwin \cite{AB} and it was also applied to the lower bound constructions for spanners and emulators by Huang and Pettie \cite{Huang2018LowerBO}. All graphs in this section are undirected, so when we refer to $\Galt^2$ we mean the graph obtained by removing the direction of every edge in $\Galt^2$. 

\subsection{Construction of the Obstacle Product Graph $\Gobs$}\label{subsec:Gobs}

In this section, we construct the graph $\Gobs = (\Vobs, \Eobs)$ from $\Galt^2$ by applying two operations: \emph{edge expansion} and \emph{clique replacement}.

\begin{itemize}
    \item[--] \textbf{Edge Expansion}: We subdivide every edge in $\Galt^2$ into a path of length $D$.
    For any edge $(v,w)$ in $\Galt^2$, we let $w_v$ be vertex in the subdivision of $(v,w)$ that is closest to $v$ i.e. there is an edge between $v$ and $w_v$. 
    
    \item[--] \textbf{Clique Replacement}: The clique replacement step is illustrated in Figure~\ref{fig:vertex_expansion}. For every vertex $v\in \Galt^2$ (i.e. not including the newly created vertices from the edge expansion step) in layer $i$ for $1\leq i\leq 2D-1$ (i.e. every layer except for the first and last), we replace $v$ with a bipartite clique $K_v$ where each side of the bipartition has size $|V_2^+(r)|$. We will refer to the edges in these bipartite cliques as \emph{clique edges}. 
    We construct the edges incident to $K_v$ as follows.
    Recall that in $\Galt^2$ the degree of $v$ to each of its adjacent layers is $|V_2^+(r)|$.
    For every edge $(w,v)$ in $\Galt^2$ where $w$ in layer $i-1$, we include an edge between $w_v$ and a vertex on the left side of $K_v$, such that each $w_v$ has an edge to a different vertex in $K_v$. Similarly, for every edge $(v,w')$ in $\Galt^2$ where $w'$ is in layer $i+1$, we include an edge between $w'_v$ and a vertex on the right side of $K_v$, such that each $w'_v$ has an edge to a different vertex in $K_v$. 
    Note that for every path of length 2 in $\Galt^2$ from a vertex $w$ in layer $i-1$, to $v$, to a vertex $w'$ in layer $i+1$, there is a unique path of length 3 in $G_{obs}$ from $w_v$ to $w'_v$ through $K_v$, and vice versa. 
   
\end{itemize}

\begin{figure}[h]
    \centering
    \includegraphics[width=10cm]{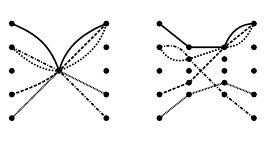}
    \caption{Left: A variety of paths traversing through a vertex in $G_{alt}^2$. Right: The corresponding paths in $G_{obs}$, after the clique replacement step. Note that not all of the existing clique edges are shown. 
    }
    \label{fig:vertex_expansion}
\end{figure}

We will now define the critical paths $\Pobs$ of $\Gobs$.
We note that because we did not perform the clique replacement step on the vertices in the first and last layers of $\Galt^2$, each such vertex uniquely corresponds to a single vertex in $\Gobs$. 
We let the critical pairs $\Pobs$ in $\Gobs$ be the pairs that uniquely correspond to the critical pairs $\Palt^2$ in $\Galt^2$. The critical path in $\Gobs$ between a pair of vertices in $\Pobs$ is defined to be exactly analogous to the critical path in $\Galt^2$ between the corresponding pair in $\Palt^2$. That is, for each edge in the critical path in $\Galt^2$ we simply take the path in $\Gobs$ created by subdividing this edge; then we connect these paths via the appropriate clique edges.

Using the bound $|V_2^+(r)| = \Theta(r^{2/3})$ from \cref{lem:polytope_vertex_num}, we can compute the quantities of $\Gobs$ as follows (the justification for each of these equations is below): 
\begin{align}
    |\Vobs| &= \Theta(D^2 (Dr)^3r^{2/3} + D(Dr)^3\cdot |V_2^+(r)) = \Theta (D^2 (Dr)^3r^{2/3})\label{eq16}\\
    |\Eobs| &= \Theta (D^2 (Dr)^3r^{2/3} +D(Dr)^3|V_2^+(r)|^2) = \Theta(D^2(Dr)^3r^{2/3}+D(Dr)^3r^{4/3})\label{eq17}\\
    |\Pobs| &=  \Theta((Dr)^3r^{4/3})\label{eq18}.
\end{align}
In \cref{eq16}, the first term counts the number of vertices due to edge expansion and the second term counts the number of vertices due to clique replacement. Specifically, in the edge expansion step, we create $D$ new vertices for every edge in $\Galt^2$ and in the clique replacement step, every vertex is replaced by $2|V_2^+(r)|$ vertices. Similarly in \cref{eq17}, the first term counts the number of edges due to edge expansion and the second term counts the number of clique edges. In the edge expansion step, an edge in $\Galt^2$ is replaced by $D$ edges and in the clique replacement step, every vertex is replaced by a clique of $\Theta(|V_2^+(r)^2)$ edges. Finally, in \cref{eq18}, notice that the the number of critical pairs remains the same as for $\Galt^2$.

\subsection{Application to Additive Spanners}\label{subsec:spanner}

Now we show how $\Gobs$ can give the lower bound stated in \cref{thm:spanner_LB}. We first recall \cref{thm:spanner_LB}: \thmspannerLB*

We first give a lemma relating the additive stretch of a spanner on $\Gobs$ with the number of clique edges it contains.

\begin{lemma}\label{lem:spanner_stretch}
Every spanner of $\Gobs$ with additive stretch at most $2D-1$ must contain at least $D|\Pobs|$ clique edges.
\end{lemma}

\begin{proof}
Assume for contradiction that there exists a spanner $H$ with additive stretch $2D-1$ that contains at most $D|\Pobs|-1$ clique edges. Since each of the $|\Pobs|$ critical paths in $\Gobs$ traverses through $2D-1$ clique edges, by the pigeonhole principle, there exists some critical path $p$ such that at least $D$ clique edges from $p$ are not included in the spanner $H$. 

Consider the sequence of bipartite cliques $K_a, K_b, K_c,\dots$ that $p$ hits, and let $p_{alt}$ be the path on the corresponding sequence of vertices $a,b,c\dots$ in $\Galt^2$. Intuitively, $p_{alt}$ is the path that $p$ becomes after ``reversing'' the obstacle product. Let $p_H$ be the shortest path in $H$ between the endpoints of $p$, and define $p_{H_{alt}}$ analogously to $p_{alt}$.

\begin{enumerate}
    \item If $p_{H_{alt}}$ is the same path as $p_{alt}$, then $p$ and $p_H$ traverse the same sequence of bipartite cliques. Recall from the definition of the clique replacement step that for every path of length 2 in $p_{alt}$ from a vertex $w$ in a layer $i-1$, to a vertex $v$ in layer $i$, to a vertex $w'$ in layer $i+1$, there is a unique path of length 3 in $p$ from $w_v$ to $w'_v$ through $K_v$. However, we know that there are $D$ clique edges on $p$ that are not included in $H$. Thus, $p_H$ cannot use these $D$ clique edges. Since the path of length 3 in $p$ from each $w_v$ to $w'_v$ is unique, $p_H$ must take a longer path from $w_v$ to $w'_v$ for each of these $D$ cliques. Since each clique is bipartite, $p_H$ must traverse at least 2 more edges than $p$ at each of these $D$ cliques. Hence $|p_H|\geq |p|+2D$.
    
    \item If $p_{H_{alt}}$ is not the same path as $p_{alt}$, then $p_{H_{alt}}$ traverses at least two more edges than $p_{alt}$ since $\Galt^2$ is bipartite. 
    Since each edge in $\Galt^2$ is replaced by a path of length $D$ in the edge expansion step, we have $|p_H|\geq |p|+2D$.
\end{enumerate}

Therefore in either case, the additive stretch is at least $2D$, which is a contradiction.
\end{proof}

Now we are ready to prove Theorem~\ref{thm:spanner_LB}.
Given a spanner $H$ of $\Gobs$ with $|E(H)|=\Theta(|\Vobs|)$,
we can set $D|\Pobs| = |E(H)|+1=\Theta(|\Vobs|)$ and then solve for $D$ and $r$ as follows. Using the bounds on $|\Pobs|$ and $|\Vobs|$ from \cref{subsec:Gobs}, we have
\[D(Dr)^3r^{4/3} = {\Theta}(D^2(Dr)^3r^{2/3}).
\]
Solving for $r$, we have $r = {\Theta}(D^{3/2})$, and therefore $|\Vobs| = n= {\Theta}(D^{21/2}) $. By Lemma~\ref{lem:spanner_stretch}, since $H$ contains less than $D|\Pobs|$ edges (and thus less than $D|\Pobs|$ clique edges), $H$ has additive stretch at least $2D= {\Omega}(n^{2/21})$.

\subsection{Application to Additive Emulators}\label{subsec:emulator}

Now we show how $\Gobs$ can give the lower bound stated in \cref{thm:emulator_LB}. First we recall \cref{thm:emulator_LB}: \thmemulatorLB*

The key property relating the additive stretch and the size of the emulator on $\Gobs$ is stated in the following lemma.

\begin{lemma}\label{lem:emulator_stretch}
Every emulator of $\Gobs$ with additive stretch at most $2D-1$ must contain at least $|\Pobs|/2$ edges.
\end{lemma}

\begin{proof}
Let $H$ be an emulator of $\Gobs$ with additive stretch at most $2D-1$. We show that we can construct a spanner $H'$ from $H$ with the same stretch (but more edges). Without loss of generality, we assume that every edge $(u,v)\in H$ has weight equal to $d_{\Gobs}(u,v)$.

For a shortest path $p$ in $H$ between a critical pair in $\Pobs$, let $p'$ be the corresponding path in $\Gobs$ where every weighted edge is replaced by a path of corresponding length. For every critical pair in $\Pobs$ with corresponding path $p$ in $H$ and $p'$ in $\Gobs$, we include $p'$ in $H'$. Notice that for any $(s,t)\in \Pobs$, we have $d_H(s,t) = d_{H'}(s,t)$. Furthermore, $H'$ is a spanner with at most $2D|E(H)|$ clique edges since every weighted edge in $H$ can contribute at most $2D$ clique edges to $H'$. Then by Lemma~\ref{lem:spanner_stretch}, $2D|E(H)|\geq D|\Pobs|$, so $|E(H)|\geq |\Pobs|/2$.
\end{proof}

Now we are ready to prove Theorem~\ref{thm:emulator_LB}.
Given an emulator $H$ of $\Gobs$ with $|E(H)|=\Theta(|\Vobs|)$,
we set $|\Pobs|/2 = |E(H)|+1=\Theta(|\Vobs|)$ and solve for $D$ and $r$ as follows. Using the bounds on $|\Pobs|$ and $|\Vobs|$ from \cref{subsec:Gobs}, we have 
\[(Dr)^3r^{4/3} = \Theta(D^2(Dr)^3r^{2/3}).
\]
Therefore we have $r = \Theta(D^3)$ and $|\Vobs| = \Theta(D^{16})$. By Lemma~\ref{lem:emulator_stretch}, since $H$ contains less than $|\Pobs|/2$ edges, $H$ has additive stretch at least $2D= {\Omega}(n^{1/16})$.

\section{Conclusion}\label{sec:conclusion}

\paragraph{Limitations on Applying Our Alternation Product to Shortcut Lower Bounds.} We will show that $\Galt^3$ gives the best possible lower bound on the $O(m)$ shortcut problem using our improved alternation product technique.
Consider generalizing the construction of $\Galt^2$ and $\Galt^3$ described in \cref{subsec:G_alt_2,subsec:G_alt_3} by taking the alternation product of more copies of $G_0$ to construct the graph $\Galt^{k-1}$ in the natural way as follows:

\begin{itemize}
    \item \textbf{Vertex Set}: The vertex set is identified with lattice points in $(\Z/(3Dr)\Z)^k$ and the graph has $(k-1)D+1$ layers.
    
    \item \textbf{Edge Set}:  If $i\equiv p \pmod{k-1}$, then we put an edge between $v_x^i, v_y^{i+1}$ with $x = (x_1,\dots, x_k), y = (y_1,\dots, y_k)$ if and only if there exists $\gamma\in V_2^+(r)$ such that $(x_{p+1}, x_{p+2})+\gamma \equiv (y_{p+1}, y_{p+2})\pmod{3Dr}$.
    
    \item \textbf{Critical Paths}: The critical paths are of the form $v_x^0, v_y^{(k-1)D}$ with $x = (x_1,\dots, x_k), y = (y_1,\dots, y_k)$, where $y$ satisfies for all $1\leq i\leq k$, $y_i = x_i + D \gamma_2^{i-1} + D\gamma_1^i$ for $(\gamma^1,\dots, \gamma^{k-1})\in (V_2^+(r))^{k-1}$ with $\gamma^i = (\gamma_1^i, \gamma_2^i)\in V^+_2(r)$. In this formula, we assume $\gamma^0 = 0, \gamma^k = 0$, so $y_1 = x_1+D\gamma_1^1$ and $y_k = x_k + D\gamma_2^{k-1}$.
\end{itemize}

Then we have 
\begin{align*}
    |V_{alt}| &= \Theta((Dr)^kD)\\
    |E_{alt}| &= \Theta((Dr)^kD|V_2^+(r)|)\\
    |P_{alt}| &= \Theta((Dr)^k|V_2^+(r)|^{k-1}) .
\end{align*}
Setting $|E_{alt}| = \Theta(|P_{alt}|)$ implies that $r = \Theta(D^{3/(2(k-2))})$ and we get
\[n = |V_{alt}| = \Theta((Dr)^kD) = \Theta\left(D^{\frac{2k^2+k-4}{2(k-2)}}\right).\]
The function $f(k) = \frac{2k^2+k-4}{2(k-2)}$ minimizes at $k = 2+\sqrt{3}\approx 3.7$ with $f(k) = \frac{9}{2}+2\sqrt{3}\approx 7.9$. Therefore, this implies that our lower bound at $k = 4$ is in fact the best we can get from this type of construction. 

\paragraph{Acknowledgements}

The authors would like to thank Gary Hoppenworth for pointing out a calculation error in the previous version of the paper.

\printbibliography

\end{document}